\documentclass[twocolumn,twoside,pra,superscriptaddress]{revtex4}
\usepackage{amsmath,amssymb,amsthm}
\usepackage[pdftex]{graphicx,color}
\usepackage{tikz,wrapfig}

\newcommand{\ket}[1]{| #1 \rangle}
\newcommand{\bra}[1]{\langle #1 |}

\newcommand{\proj}[1]{| #1 \rangle\langle #1 |}

\newcommand{\spinup}{A^{\dagger}}
\newcommand{\spindown}{A}

\newtheorem{theorem}{Theorem}
\newtheorem{corollary}{Corollary}

\newtheorem{defin}{Definition}
\newtheorem{requirement}{Requirement}
\newtheorem{lemma}{Lemma}

\numberwithin{equation}{section}

\begin{document}
\title{Causal Fermions in Discrete Spacetime}

\author{Terence C.\ Farrelly} \email{tcf24@cam.ac.uk} \affiliation{DAMTP, Centre for Mathematical Sciences, Wilberforce Road, Cambridge, CB3 0WA, United Kingdom}
\author{Anthony J.~Short}\email{tony.short@bristol.ac.uk}\affiliation{H. H. Wills Physics Laboratory, University of Bristol$\text{,}$ Tyndall Avenue, Bristol, BS8 1TL, United Kingdom}

\begin{abstract}
In this paper, we consider fermionic systems in discrete spacetime evolving with a strict notion of causality, meaning they evolve unitarily and with a bounded propagation speed.  First, we show that the evolution of these systems has a natural decomposition into a product of local unitaries, which also holds if we include bosons.  Next, we show that causal evolution of fermions in discrete spacetime can also be viewed as the causal evolution of a lattice of qubits, meaning these systems can be viewed as quantum cellular automata.  Following this, we discuss some examples of causal fermionic models in discrete spacetime that become interesting physical systems in the continuum limit: Dirac fermions in one and three spatial dimensions, Dirac fields and briefly the Thirring model.  Finally, we show that the dynamics of causal fermions in discrete spacetime can be efficiently simulated on a quantum computer.
\end{abstract} 

\maketitle
\section{Introduction}
\label{sec:Introduction}
The idea that there is a maximum speed of propagation of information has become one of the fundamental principles of physics.  In particular, it appears in relativistic quantum field theories, the most interesting example of which is the standard model, which provides a unified framework for describing the effects of the strong, weak and electromagnetic forces.  In this paper we are going to look at discrete spacetime quantum systems with the requirement of strict causality.  Quantum systems in discrete space but continuous time with local Hamiltonians do not generally satisfy this notion of causality.  (For an example illustrating why local Hamiltonians in discrete space typically lead to instantaneous propagation of information, see appendix \ref{sec:Continuous Time Models in Discrete Space}.)  There are a few reasons that causal quantum systems in discrete spacetime are interesting.

The first is simulation: simulations of physics usually start by discretizing continuous degrees of freedom.  Of particular interest to us is the simulation of relativistic quantum field theories by a quantum computer.  Simulation of quantum physics is likely to be the first practical application where a quantum computer could out-preform a classical one \cite{Lloyd96}.  The idea that quantum computers may be better equipped than classical computers to simulate quantum systems dates back at least as far as Feynman in \cite{Feynman82}.  The basic idea was that quantum systems evolving in a local way ought to be efficiently simulable on a quantum computer using local operations, something that appears impossible on a classical computer.  This was shown to be true for continuous time quantum systems with local Hamiltonians in \cite{Lloyd96} and, for fermions in particular, in \cite{AL97}.  These approaches rely upon breaking up the total Hamiltonian $H$ into a sum of $k$ body terms $H_{l}$, where $k$ does not grow with the system size, and using a Trotter expansion in terms of $e^{-iH_{l}t/n}$ to recover $e^{-iHt}$.  For causal unitaries, however, it is not obvious how to simulate the evolution efficiently.  We show how to do this for fermions by decomposing the total unitary evolution operator exactly into local unitaries on qubits, such that the evolution can be efficiently simulated on a quantum computer.

The second reason we may be interested in discrete spacetime quantum models that are strictly causal is that nature itself may be discrete, and it is plausible that even at the smallest length scales there is a strict form of causality.  The idea that the spacetime continuum breaks down in some way at very small length scales is often embraced with a view to constructing a theory reconciling gravity with quantum mechanics, such as causal sets \cite{BLMS87}.  What is really interesting is that there are causal discrete spacetime models that become interesting relativistic quantum field theories in the continuum limit.  We describe some of these in section \ref{sec:Constructing Models}.  It is encouraging that, although these discrete systems do not have Lorentz symmetry, it is recovered in the continuum limit.  Currently, the only such models are fermionic; a particularly interesting example is given in \cite{DdV87}, which becomes the Thirring model in two dimensional spacetime.  Ideally, one would like a general recipe giving a discrete spacetime model that converges to a given quantum field theory in the continuum limit if it exists.  Taking such limits, however, is a complicated process as couplings must be renormalized.

The final reason causal quantum models in discrete spacetime are interesting is that they provide discretized models of relativistic systems that are well defined.  Despite its remarkable successes in explaining physics at very high energy scales, quantum field theory in the continuum has yet to be put on a firm mathematical footing.  Here, because we work with a discrete lattice, we have a regulator, so the infinities that plague quantum field theory do not appear.

The models we will look at are quantum systems in discrete spacetime with a strict notion of causality.  They evolve over each timestep via a causal unitary operator, which essentially means that localized observables can only spread a finite amount over one timestep.  Our main focus will be fermionic systems: one of our results is the proof of a general principle (theorem \ref{th:1} in section \ref{sec:Local Decomposition}) that says that causal fermionic evolution can be decomposed exactly into a product of local commuting fermionic unitaries, which is analogous to a result of \cite{ANW11} for quantum cellular automata.  This extends their maxim that \emph{``unitarity plus causality implies localizability''} to systems of fermions.  We also extend this to systems of fermions and bosons that may be interacting.

Causal fermionic systems in discrete spacetime are fermionic analogues of {\it Quantum Cellular Automata} (QCA) \cite{ANW11,GNVW12,SW04,Watrous95}.  These are discrete spatial lattices that have finite quantum systems (with associated finite dimensional Hilbert spaces) at each spatial point.  QCA evolve in discrete time via a causal unitary operator \footnote{Usually, QCA are defined over infinite lattices, so the evolution is an isomorphism of the C*-algebra of quasi local operators \cite{SW04}.  See appendix \ref{sec:Infinite Systems} for more details.}.  A useful picture to have in mind is a spin lattice that evolves over discrete timesteps in a causal way.  QCA are interesting for many reasons.  They are universal for quantum computation \cite{Watrous95}, with the particularly nice property that they are implementable by applying local unitary gates \cite{ANW11,GNVW12}.  We show that the dynamics of causal fermionic systems in discrete spacetime can be viewed as subsectors of the dynamics of QCA (theorem \ref{th:2}).  This is analogous to a result in \cite{Ball05,VC05}, which maps local fermionic Hamiltonians to local spin Hamiltonians.  Their results for local Hamiltonians provide a way of thinking about fermions without the need for anticommuting operators on different sites.  Theorem \ref{th:2} extends this to systems of causal fermions in discrete spacetime.  It is interesting that in both cases anticommuting fermionic operators are not necessary to describe the dynamics.

The breakdown of this paper is as follows.  We start in section \ref{sec:Fermions and the Jordan-Wigner Transformation} by going through the properties of fermions and the Jordan-Wigner Transformation, which allows us to represent fermionic operators by operators on qubits (two dimensional quantum systems, which we can think of as spin $1/2$ particles).  In section \ref{sec:Causal Fermions}, we look at causal fermionic evolution in discrete spacetime and show that it always has a decomposition in terms of local fermionic unitaries.  We also point out that this extends to possibly interacting bosonic and fermionic modes.  In section \ref{sec:Representation by Qubits}, we show that we can view any causal fermionic evolution as a causal discrete time evolution of qubits, or equivalently as a subsector of the evolution of a quantum cellular automaton.  Next, in section \ref{sec:Constructing Models}, we give some discrete spacetime models that become interesting systems in the continuum limit: first, we reproduce the evolution given in \cite{DdV87,D'Ariano12a,D'Ariano12b} of discrete spacetime Dirac fermions in one dimensional space and its corresponding local implementation on qubits.  Second, we look at a similar discrete fermionic model that behave like Dirac fermions in three dimensional space in the continuum limit, originally given for a single particle in \cite{Bial94}.  In section \ref{sec:Fermionic Fields}, we look at the representation of the Dirac field in discrete spacetime.  This is followed by a discussion in section \ref{sec:Simulating Causal Fermions} of the implications of these results for simulation.
After concluding remarks, we extend these results to infinite lattices via C*-algebras in the appendix.

Note that we set $\hbar=c=1$ throughout.

\section{Fermions and the Jordan-Wigner Transformation}
\label{sec:Fermions and the Jordan-Wigner Transformation}
For simplicity, here we will restrict the set of spatial points to be a finite-sized $d$-dimensional lattice, so that positions are labelled by vectors of integers. To introduce translational symmetry, it will sometimes be helpful to treat this lattice as a torus (i.e. to introduce periodic boundary conditions). The results all have natural extensions to different geometries, but it will be useful to have this particular example in mind for when we discuss locality and causality.  We postpone a discussion of systems on infinite lattices to appendix \ref{sec:Infinite Systems}, as these involve additional complications that are not necessary to understand the main ideas.

Now suppose that we have fermionic modes at each point.  A natural example of this is a system where each site can be occupied by spin up or spin down electrons.  We denote the state with all modes unoccupied by $\ket{\Omega}$.  Then we define creation and annihilation operators $a^{\dagger}_{\vec{x}\mu}$ and $a_{\vec{x}\mu}$, where $\vec{x}$ labels the position and $\mu$ labels the extra degree of freedom at each site.  These satisfy the canonical anticommutation relations:
\begin{equation}
 \begin{split}
 \{a^{\dagger}_{\vec{x}\mu},a_{\vec{y}\nu}\} & =\delta_{\mu\nu}\delta_{\vec{x}\vec{y}}\\
 \{a_{\vec{x}\mu},a_{\vec{y}\nu}\} & =0,
\end{split}
\end{equation}
where $\delta_{\vec{x}\vec{y}}=1$ if $\vec{x}=\vec{y}$ and is zero otherwise.  We also have that $a_{\vec{x}\mu}\ket{\Omega}=0$, and the Hilbert space is spanned by all possible products of creation operators $a^{\dagger}_{\vec{x}\mu}$ acting on $\ket{\Omega}$.  For example, the state $a^{\dagger}_{\vec{x}\mu}a^{\dagger}_{\vec{y}\nu}\ket{\Omega}$ has a fermion at $\vec{x}$ and a fermion at $\vec{y}$, with internal degrees of freedom $\mu$ and $\nu$ respectively.

An extra requirement that we make of systems of fermions is that physical observables are not only self-adjoint but are linear combinations of products of even numbers of creation and annihilation operators.  In particular, in continuous time systems, all physical Hamiltonians have this property.  For example, the Hubbard Hamiltonian is
\begin{equation}
\label{eq:4}
\begin{split}
H=- & \alpha\sum_{<\vec{x}\vec{y}>} \sum_{\mu} (a^{\dagger}_{\vec{x}\mu}a^{\ }_{\vec{y}\mu}+a^{\dagger}_{\vec{y}\mu}a^{\ }_{\vec{x}\mu})+\\
 & U\sum_{\vec{x}}(a^{\dagger}_{\vec{x}\uparrow}a^{\ }_{\vec{x}\uparrow})(a^{\dagger}_{\vec{x}\downarrow}a^{\ }_{\vec{x}\downarrow}),
\end{split}
\end{equation}
where $\mu\in\{\uparrow,\downarrow\}$ labels the extra degree of freedom (spin in this case), $\langle \vec{x}\vec{y}\rangle$ denotes nearest neighbour pairs, and $\alpha,U\geq 0$ are real parameters. The first term describes electrons hopping, while the second term is an on site Coulomb repulsion.

Because fermion creation and annihilation operators anticommute regardless of the spatial separation between them, these operators are, in a sense, non local objects.  If we want to represent a fermionic system by a system of qubits, we will have to take this into account.  For example, suppose we want to represent a line of $N$ fermionic modes with no internal degree of freedom and positions labelled by $x\in\{0,1,...,N-1\}$ by $N$ qubits.  It is natural to take $\ket{00...0}$ to represent the state with no fermions present.  Now we can choose the representation of $a^{\dagger}_0$ on the qubits to be
\begin{equation} 
\spinup_0=\textstyle\frac{1}{2}(X_0-iY_0) = \ket{1}_0\bra{0},
\end{equation} where $X$, $Y$ and $Z$ are the standard Pauli operators \footnote{$X=\ket{0}\bra{1} + \ket{1}\bra{0}, Y=i(\ket{1}\bra{0} - \ket{0}\bra{1}), Z=\proj{0}- \proj{1}$.}.    The subscript $0$ implies that these operators act on all other qubits like the identity, meaning, for example, $X_0=X\otimes I\otimes I...\otimes I$. Because $a^{\dagger}_0$ and $a^{\dagger}_1$ anticommute, we cannot simply pick $a^{\dagger}_1$ to be represented by $\spinup_1$.  Instead, we can satisfy the anticommutation relations by taking
\begin{equation}
a^{\dagger}_x\equiv \spinup_x\prod_{y<x}Z_y. \label{eq:ordering} 
\end{equation}
This is known as the Jordan-Wigner Transformation \cite{JW28}.  It preserves the anticommutation relations because of the strings of $Z$s.  In fact, we are free to choose the ordering in the product above however we want.  To see this, we give a more general Jordan-Wigner Transformation, which is particularly useful for higher dimensional lattices.  Given $N$ fermionic modes, we associate a qubit to each mode.  Next, we assign a unique label to each site, $\pi(\vec{x})\in\{0,...,N-1\}$, and define
\begin{equation}
\label{eq:general ordering}
a^{\dagger}_{\vec{x}}\equiv \spinup_{\vec{x}}\hspace{-1em}\prod_{\pi(\vec{y})<\pi(\vec{x})}\hspace{-1em}Z_{\vec{y}},
\end{equation}
which also satisfies the anticommutation relations.  For the special case of a line of fermions with $\pi(x)=x$ we recover equation \ref{eq:ordering}.  It is sometimes useful to use different ordering schemes for different problems.

If there are multiple fermionic modes at a lattice site (for example, due to spin, or different types of particle), then the qubit representation will include a separate qubit at that lattice site for each mode. Then the ordering $\pi(\vec{x}, \mu)$, where $\mu$ distinguishes different modes at site $\vec{x}$, will assign a unique number to each mode. It is natural to choose the ordering such that $\pi(\vec{x}, \mu)$ for the set of modes at each site are consecutive.  This means that even products of fermionic creation and annihilation operators at the same site are local in the qubit representation.

With the ordering in (\ref{eq:ordering}), local fermionic operators on a line are represented by local operators on the corresponding line of qubits.  In higher spatial dimensions, however, it is not generally true that local fermionic operators correspond to local qubit operators.  In fact, even for a ring of fermions this is not generally true.  Luckily, there is a way of getting around this that involves introducing auxiliary fermions, given in \cite{Ball05,VC05} and reproduced in appendix \ref{sec:Representing Local Fermionic Hamiltonians by Local Qubit Hamiltonians}, which we will use to extend the results of the following sections to systems of fermions in arbitrary spatial dimensions.

\section{Causal Fermions}
\label{sec:Causal Fermions}
In the following sections, where we prove our main results, we look at discrete time systems.  In continuous time, it is natural to work with a Hamiltonian as it determines the evolution via the Schr\"{o}dinger equation.  In a discrete time picture, however, there is no Schr\"{o}dinger equation, so it is more natural to work directly with the unitary operator that acts on the state each timestep.

To define notions of locality and causality, it is helpful to define the \emph{neighbourhood} of a spatial point. Here we take this to mean the set of points that are at most 1 lattice step away in each spatial direction (taking into account the periodic boundary conditions if necessary) \footnote{This definition will be convenient for the examples we consider in section \ref{sec:Constructing Models}.}. The neighbourhood of a point is therefore a $d$-dimensional hypercube of side length 3, centered on that point. Note that this definition of the neighbourhood is not critical to the proofs: for example, if there were next nearest neighbour interactions, we could consider larger hypercubes.

We say that a fermionic operator is localized on a spatial region $R$ if it can be written in terms of creation and annihilation operators corresponding only to $R$.  Because of this, and the fact that they must be sums of even products of creation and annihilation operators, localized observables from non overlapping regions of space always commute.

Next, we define a causal fermionic unitary.  Note that we work in the Heisenberg picture.
\begin{defin}
 A fermionic unitary $U$ is \emph{causal} if, for any $\vec{x}$ and $\mu$, $U^{\dagger}a_{\vec{x}\mu}U$ is localized in the neighbourhood of $\vec{x}$.
\end{defin}
The definition of a causal unitary ensures that over one timestep information cannot propagate farther than one step in each spatial direction.  Note that we can construct any operator localized on a region from creation and annihilation operators corresponding to that region.  So, in particular, this definition ensures that local observables do not spread very far after one timestep.

In continuous time systems, the Hamiltonian of a system of fermions is a sum of even products of creation and annihilation operators.  In particular, this implies that $e^{-iHt}$ commutes with the annihilation operator $b$ if $H$ does not contain any terms with $b$ or $b^{\dagger}$.  We take it for granted that the discrete time dynamics of fermions we consider also have this property.  This means that, given a system of fermions evolving via $U$, we can add additional fermionic modes whose creation operators anticommute with the original fermion creation and annihilation operators while commuting with $U$.  (We could just assume that $U=e^{iA}$, where $A$ is a sum of even products of creation and annihilation operators, but this is less general and not useful for infinite systems.)

Finally, we have the following useful lemma, which is proved in appendix \ref{sec:Proofs of Lemmas 1 and 2}.

\begin{lemma}
\label{lem:1}
The inverse of a causal fermionic unitary is also a causal fermionic unitary.
\end{lemma}

\subsection{Local Decomposition}
\label{sec:Local Decomposition}
Here we will give a local unitary decomposition for causal fermionic unitaries analogous to that given for QCA in \cite{ANW11,GNVW12}. In short, we will show that the full unitary evolution can be decomposed into a product of unitaries, each of which is localised in a particular neighbourhood.

Constructing this local decomposition requires us to consider the joint evolution of the system of fermions together with an identical copy of that system.  Let us denote the creation operators for the original system of fermions by $a^{\dagger}_{\vec{x}\mu}$ and the creation operators for the corresponding modes of a copy system by $b^{\dagger}_{\vec{x}\mu}$.

It will be useful to define a fermionic swap operator.  A unitary $\mathbf{S}_{\vec{x}\mu}$ implementing the fermionic swap $a^{\dagger}_{\vec{x}\mu}\leftrightarrow b^{\dagger}_{\vec{x}\mu}$, meaning $\mathbf{S}_{\vec{x}\mu}a_{\vec{x}\mu}\mathbf{S}_{\vec{x}\mu}=b_{\vec{x}\mu}$ and $\mathbf{S}_{\vec{x}\mu}b_{\vec{x}\mu}\mathbf{S}_{\vec{x}\mu}=a_{\vec{x}\mu}$, is
\begin{equation}
\mathbf{S}_{\vec{x}\mu}= \exp[i\frac{\pi}{2}(b^{\dagger}_{\vec{x}\mu}-a^{\dagger}_{\vec{x}\mu})(b_{\vec{x}\mu}-a_{\vec{x}\mu})].
\end{equation}
See appendix \ref{app:swap} for more details about this operator.

The local decomposition of a causal evolution of fermions is derived in the following theorem.

\begin{theorem}
\label{th:1}
Given a finite system of fermions with creation operators $a^{\dagger}_{\vec{x}\mu}$, evolving via the causal unitary $U_A$, the evolution of two copies of this system via $U_AU_B^{\dagger}$, where $U_B$ is equivalent to $U_A$ but acting on the $b^{\dagger}_{\vec{x}\mu}$ fermions, can be decomposed into local fermionic unitaries:
\begin{equation}
 U_AU_B^{\dagger}=\prod_{\vec{x},\mu}(\mathbf{S}_{\vec{x}\mu})\prod_{\vec{y},\nu}[U_B\mathbf{S}_{\vec{y}\nu}U_B^{\dagger}],
\end{equation}
where $U_B\mathbf{S}_{\vec{y}\nu}U_B^{\dagger}$ are commuting local fermionic unitaries.
\end{theorem}
\begin{proof}
First,
\begin{equation}
 \prod_{\vec{x},\mu}(\mathbf{S}_{\vec{x}\mu})\prod_{\vec{y},\nu}[U_B\mathbf{S}_{\vec{y}\nu}U_B^{\dagger}]=\mathbf{S}U_B\mathbf{S}U_B^{\dagger},
\end{equation}
where
\begin{equation}
\mathbf{S}=\prod_{\vec{x},\mu}(\mathbf{S}_{\vec{x}\mu})
\end{equation}
is the unitary that swaps all modes.  And, because $\mathbf{S}$ swaps all modes, $\mathbf{S}U_B\mathbf{S}=U_A$.  It follows that
\begin{equation}
 \prod_{\vec{x},\mu}(\mathbf{S}_{\vec{x}\mu})\prod_{\vec{y},\nu}[U_B\mathbf{S}_{\vec{y}\nu}U_B^{\dagger}]=U_AU_B^{\dagger}.
\end{equation}
Furthermore, $U_B\mathbf{S}_{\vec{x}\mu}U_B^{\dagger}$ is a local fermionic unitary because $\mathbf{S}_{\vec{x}\mu}$ is just
\begin{equation}
 \exp[i\frac{\pi}{2}(b_{\vec{x}\mu}^{\dagger}-a_{\vec{x}\mu}^{\dagger})(b_{\vec{x}\mu}-a_{\vec{x}\mu})],
\end{equation}
and so
\begin{equation}
 U_B\mathbf{S}_{\vec{x}\mu}U_B^{\dagger}=\exp[i\frac{\pi}{2}(b_{\vec{x}\mu}^{\prime\dagger}-a_{\vec{x}\mu}^{\dagger})(b^{\prime}_{\vec{x}\mu}-a_{\vec{x}\mu})],
\end{equation}
where $b^{\prime}_{\vec{x}\mu}=U_Bb_{\vec{x}\mu}U_B^{\dagger}$, which must be localized within the neighbourhood of $\vec{x}$ because $U_B^{\dagger}$ is causal. Hence $ U_B\mathbf{S}_{\vec{x}\mu}U_B^{\dagger}$ is also localised within the neighbourhood of $\vec{x}$.  (Naturally, we are thinking of each mode labelled by $\vec{x}$ as being at the same site as its copy, which is why the unitaries $U_B\mathbf{S}_{\vec{x}\mu}U_B^{\dagger}$ are local.) The fact that the different unitaries $ U_B\mathbf{S}_{\vec{x}\mu}U_B^{\dagger}$ commute follows from $[\mathbf{S}_{\vec{x}\mu}, \mathbf{S}_{\vec{y}\nu}]=0$. 
\end{proof}
This theorem tells us that the joint causal evolution of two copies of a system of fermions can be decomposed into a product of local unitaries.  Now note that, given any state $\ket{\psi}$ of the physical system of fermions and its copy, it follows that, for any measurement operator $\mathcal{M}_A$ on the physical fermions,
\begin{equation}
 \bra{\psi}U_BU_A^{\dagger}\mathcal{M}_AU_AU_B^{\dagger}\ket{\psi}=\bra{\psi}U_A^{\dagger}\mathcal{M}_AU_A\ket{\psi},
\end{equation}
so this joint evolution of both the physical and auxiliary fermions is just as good as the original evolution but with the advantage of being decomposable into local fermionic unitaries.  For example, $\ket{\psi}$ could be any state of the physical fermions with all auxiliary modes unoccupied, such as $\textstyle{\frac{1}{\sqrt{2}}(a^{\dagger}_{\vec{x}\mu}+a^{\dagger}_{\vec{y}\nu}})\ket{\Omega}$, where $\ket{\Omega}$ is the state annihilated by all $a_{\vec{x}\mu}$ and $b_{\vec{x}\mu}$.

In some cases there is a natural local unitary decomposition of the evolution without the need to include a copy of the system: see the example in section \ref{sec:Discrete Dirac Fermions in One Dimension}.  In general, however, this is not true.  A simple counterexample is the unitary that shifts everything one step to the right every timestep.

As an aside, note that a similar decomposition exists for systems of bosons and fermions, which may be interacting.  To make the notation simple, we suppose that there is only one fermionic and one bosonic mode at each point, though the result still holds with multiple modes at each site.  We denote by  $\mathbb{S}_{\vec{x}}$ the bosonic swap operator between the mode at $\vec{x}$ with creation operator $c^{\dagger}_{\vec{x}}$ and its copy with creation operator $d^{\dagger}_{\vec{x}}$.
\begin{equation}
\mathbb{S}_{\vec{x}}= \exp[i\frac{\pi}{2}(d^{\dagger}_{\vec{x}}-c^{\dagger}_{\vec{x}})(d_{\vec{x}}-c_{\vec{x}})].
\end{equation}

The local decomposition is encapsulated in the following corollary.

\begin{corollary}
\label{cor:1}
 Given a finite system of fermions and bosons evolving via the causal unitary $U_A$, the evolution of two copies of this system via $U_AU_B^{\dagger}$, where $U_B$ is equivalent to $U_A$ but acting on the copy system, can be decomposed into local unitaries:
\begin{equation}
\begin{split}
 & U_AU_B^{\dagger}\\
& =\prod_{\vec{w}}(\mathbf{S}_{\vec{w}})\prod_{\vec{x}}(\mathbb{S}_{\vec{x}})\prod_{\vec{y}}[U_B\mathbf{S}_{\vec{y}}U_B^{\dagger}]\prod_{\vec{z}}[U_B\mathbb{S}_{\vec{z}}U_B^{\dagger}],
\end{split}
\end{equation}
where $U_B\mathbf{S}_{\vec{x}}U_B^{\dagger}$ and $U_B\mathbb{S}_{\vec{x}}U_B^{\dagger}$ are commuting local unitaries.
\end{corollary}

\subsection{Representation by Qubits}
\label{sec:Representation by Qubits}
Now that we have a local decomposition of a causal fermionic unitary,
we can look at its representation in the qubit picture.  In this and
the following section, we will assume there is only one fermionic mode
per site to make the notation simpler.  (The extension to more than
one mode per site is straightforward.)  We assign qubits to the fermionic mode at $\vec{x}$ and its copy (created by $a^{\dagger}_{\vec{x}}$ and $b^{\dagger}_{\vec{x}}$ respectively).  For now we will work with a finite line of points, so that with the natural choice of ordering for the Jordan-Wigner Transform (equation \ref{eq:ordering}) the operators $\mathbf{S}_{\vec{x}}$ are local unitaries in the qubit representation.  Next, recall that $U_B\mathbf{S}_{\vec{x}}U_B^{\dagger}$ is just
\begin{equation}
\exp[i\frac{\pi}{2}(b^{\prime\dagger}_{\vec{x}}-a^{\dagger}_{\vec{x}})(b^{\prime}_{\vec{x}}-a_{\vec{x}})].
\end{equation}
Now note that $b^{\prime\dagger}_{\vec{x}}=U_Bb^{\dagger}_{\vec{x}}U_B^{\dagger}$ must be a linear combination of {\it odd} powers of creation and annihilation operators on the neighbourhood of $\vec{x}$.  (This is proved in lemma \ref{lem:2} in appendix \ref{sec:Proofs of Lemmas 1 and 2}.)  So $U_B\mathbf{S}_{\vec{x}}U_B^{\dagger}$ has the form $e^{-iH_{\vec{x}}}$, where $H_{\vec{x}}$ is a self adjoint operator localized on the neighbourhood of $\vec{x}$ containing only even products of creation and annihilation operators.  But with the natural ordering for the Jordan-Wigner Transformation in equation \ref{eq:ordering}, $H_{\vec{x}}$ is localised on the neighbourhood of $\vec{x}$ in the qubit representation also.  This means that in the qubit representation $U_B\mathbf{S}_{\vec{x}}U_B^{\dagger}$ is localised on the neighbourhood of $\vec{x}$.

Hence, we can view this causal evolution of fermions as a causal evolution of qubits.  In higher spatial dimensions or lines with periodic boundary conditions, we need to do more to ensure that our causal fermionic evolution can be represented by local unitaries acting on qubits.  We see how to do this in the next section.  (Note that the main points about simulation and models of quantum field theories can be understood without going through the details of this.)  This proves the following theorem.

\begin{theorem}
\label{th:2}
 Any causal fermionic evolution in discrete spacetime is equivalent to a subsector of the causal evolution of a system of qubits, which is a type of quantum cellular automaton.
\end{theorem}

\subsection{More than one spatial dimension}
\label{sec:More than one spatial dimension}
Now that we are considering higher spatial dimensions, in general there is no choice of ordering for the Jordan-Wigner Transformation such that all local unitaries in the fermion picture are local in the qubit picture.  Naturally, we choose the ordering so that fermionic modes and their copies, which are associated to the same site, are consecutive in the ordering scheme.  This means that each fermionic swap $\mathbf{S}_{\vec{x}}$ term in the decomposition is still local in the qubit representation.

The $U_B\mathbf{S}_{\vec{x}}U_B^{\dagger}$ terms in the decomposition are more problematic.  They are not necessarily localized unitaries in the qubit picture.  Fortunately, we can circumvent this problem but at the price of adding auxiliary fermions, which also means that we need more qubits to represent this larger system of fermions.

First, we already know that $U_B\mathbf{S}_{\vec{x}}U_B^{\dagger}$ is a local fermionic unitary.  And, as we saw in the previous section, it has the form $e^{-iH_{\vec{x}}}$, where $H_{\vec{x}}$ is a self adjoint operator localized on the neighbourhood of $\vec{x}$ containing only even products of creation and annihilation operators.  Next, we use an idea from \cite{Ball05,VC05}.  (For brevity, here we will just give a rough idea of how this works; see section \ref{sec:Representing Local Fermionic Hamiltonians by Local Qubit Hamiltonians} for full details.)  Suppose, for example, that $H_{\vec{x}}$ contains the term $a_{\vec{x}}b_{\vec{y}}$, where $\vec{y}$ is in the neighbourhood of $\vec{x}$.  Suppose also that this term is non local in the qubit representation: this will be because of strings of $Z$ operators arising from the Jordan-Wigner Transformation.  Now we introduce two new fermionic modes, one at $\vec{x}$ and one at $\vec{y}$, with annihilation operators $c_{\vec{x}}$ and $c_{\vec{y}}$.  Because we choose the Jordan-Wigner Transformation ordering such that modes at the same site are consecutive, $a_{\vec{x}}c_{\vec{x}}$ and $c_{\vec{y}}b_{\vec{y}}$ are local in the qubit picture.  This is because all $Z$ operators in the qubit representation of $a_{\vec{x}}$ corresponding to sites other than $\vec{x}$ are cancelled by those from the qubit representation of $c_{\vec{x}}$.

Next, we make the replacement
\begin{equation}
\label{eq:6}
 a_{\vec{x}}b_{\vec{y}}\rightarrow a_{\vec{x}}(im_{\vec{x}}m_{\vec{y}})b_{\vec{y}},
\end{equation}
where
\begin{equation}
\begin{split}
 m_{\vec{x}} & =c_{\vec{x}}+c^{\dagger}_{\vec{x}}\\
 m_{\vec{y}} & =c_{\vec{y}}+c^{\dagger}_{\vec{y}},
\end{split}
\end{equation}
which we can think of as Majorana fermions.  The new term on the right hand side of \ref{eq:6} is local in the qubit picture, and, acting on a $+1$ eigenstate of $im_{\vec{x}}m_{\vec{y}}$,
\begin{equation}
 a_{\vec{x}}(im_{\vec{x}}m_{\vec{y}})b_{\vec{y}} = a_{\vec{x}}b_{\vec{y}}.
\end{equation}
Similarly, for any other terms in $H_{\vec{x}}$ that are non local in the qubit representation, we can add additional fermions, so that, acting on $+1$ eigenstates of all the additional terms like $im_{\vec{x}}m_{\vec{y}}$ that we add, the resulting operator is equal to $H_{\vec{x}}$.

This means that $U_B\mathbf{S}_{\vec{x}}U_B^{\dagger}$ is equivalent to a local fermionic unitary on a larger system that is local in the qubit picture.  So we have extended theorem \ref{th:2} to any dimension.

Note that the fermionic unitaries $U_B\mathbf{S}_{\vec{x}}U_B^{\dagger}$ commute.  After introducing additional fermions, the new unitaries $V_{\vec{x}}$ implementing $U_B\mathbf{S}_{\vec{x}}U_B^{\dagger}$ on the qubits do not necessarily commute when the neighbourhoods on which they are localized overlap.  Notice, however, that the order in which they are applied does not matter when acting on a $+1$ eigenstate of the Majorana pairs $im_{\vec{x}}m_{\vec{y}}$.  Furthermore, we can apply many $V_{\vec{x}}$ simultaneously, provided they act on non overlapping regions.  For example, for a line, we can apply all $V_x$ with $x\bmod 3=0$ first, followed by all $V_x$ with $x\bmod 3=1$, followed by all $V_x$ with $x\bmod 3=2$.  So for a line we need only three steps to implement every $V_x$.

On one hand, it seems that we could have just expanded the causal fermionic unitary $U$ as a sum of products of creation and annihilation operators and applied the trick of adding Majorana fermions directly to this.  It is not clear, however, that we can do this and preserve unitarity.  Applying the prescription to a Hamiltonian worked because we could always preserve self-adjointness.  Instead, our approach was to derive an exact decomposition of causal fermionic unitaries into a product of local unitaries, each of which can be written in the form $e^{-iH_{\vec{x}}}$, where each $H_{\vec{x}}$ has support only on the neighbourhood of $\vec{x}$.  We then applied the trick of adding Majorana fermions to each $H_{\vec{x}}$.

It is important from the point of view of simulation and for the extension to infinite lattices that the number of additional fermions we need to add per site does not depend on $\mathcal{N}$, the number of original fermionic modes.

\section{Constructing Models}
\label{sec:Constructing Models}
In the following sections, we take a constructive approach and examine specific examples of causal discrete spacetime models that become interesting continuum models as we take a continuum limit.  We also represent the evolution of such discrete spacetime systems by products of local unitaries on qubits.  Note that the realization of this for the one dimensional Dirac equation appeared in \cite{D'Ariano12a,D'Ariano12b,BDT12}.

\subsection{Discrete Dirac Fermions in One Dimension}
\label{sec:Discrete Dirac Fermions in One Dimension}
In this section, we will look at fermions in discrete spacetime that obey the one dimensional Dirac equation in the continuum limit.  We will use theorem \ref{th:1} to decompose the evolution into a simple product of local fermionic unitaries, and then we will give a representation of this in terms of qubits.  Both of these steps allow us to reproduce the evolution given in \cite{D'Ariano12a,D'Ariano12b,BDT12}.  (A single particle in discrete spacetime that obeys the Dirac equation in the continuum limit appeared in \cite{FH65}.  The details of this continuum limit were studied further in \cite{Strauch06,Strauch07,BES07}.)

For simplicity, suppose that our discrete space is finite, with sites labelled by $n\in\{0,...,N-1\}$ with periodic boundary conditions.  And suppose there are two fermionic modes at each site, labelled by $l$ and $r$.  As we anticipate getting Dirac fermions in the continuum limit, let us denote creation operators by $\psi^{\dagger}_{n,a}$, where $n\in\{0,...,N-1\}$ and $a\in\{l,r\}$.  Also define
\begin{equation}
 \psi_{n}=\begin{pmatrix}
           \psi_{n,r}\\
           \psi_{n,l}
          \end{pmatrix}.
\end{equation}
We define the matrices $\beta$ and $\alpha_1$ in this basis to be
\begin{equation}
 \beta=\begin{pmatrix}
            0 & 1\\
            1 & 0\\
           \end{pmatrix},\ \ \alpha_1=\begin{pmatrix}
            1 & 0\\
            0 & -1\\
           \end{pmatrix}.
\end{equation}

It will be convenient to work in momentum space: $\mathbf{p}=2\pi k/N$ is the discrete momentum, where $k\in\{-\textstyle\frac{N-1}{2},...,\textstyle\frac{N-1}{2}\}$, and we take $N$ to be odd.  The momentum creation operators are
\begin{equation}
\psi^{\dagger}_{\mathbf{p},a} =\frac{1}{\sqrt{N}}\displaystyle\sum_{n}e^{i\mathbf{p}n}\psi^{\dagger}_{n,a}.
\end{equation}

Let us also suppose that over each timestep the system evolves via the unitary $U=WT$.  In the continuum limit, the unitary $W$ will contribute the mass term in the Hamiltonian, and the unitary $T$ will contribute the momentum term.  First, $T$ is a conditional shift that has the effect
\begin{equation}
\label{eq:3}
T\psi_{n}T^{\dagger}=\begin{pmatrix}
           T\psi_{n,r}T^{\dagger}\\
           T\psi_{n,l}T^{\dagger}
          \end{pmatrix}=
\begin{pmatrix}
\psi_{n+1,r}\\
\psi_{n-1,l}
\end{pmatrix}.
\end{equation}
We can write $T$ in terms of the discrete momentum operator.  The operator that translates $\psi_{n,a}$ one step to the right is
\begin{equation}
 \exp(-iP_{a}),
\end{equation}
where
\begin{equation}
P_{a} =\displaystyle\sum_{\mathbf{p}}\mathbf{p}\,\psi^{\dagger}_{\mathbf{p},a}\psi_{\mathbf{p},a}.
\end{equation}
So we have
\begin{equation}
T=\exp(-i[P_r-P_l])=\exp(-i\displaystyle\sum_{\mathbf{p}}\mathbf{p}\psi^{\dagger}_{\mathbf{p}}\alpha_1\psi_{\mathbf{p}}).
\end{equation}

We define $W$ to be
\begin{equation}
W=\exp(-iM\displaystyle\sum_{\mathbf{p}}\psi^{\dagger}_\mathbf{p}\beta\psi_\mathbf{p}).
\end{equation}

Now, to take a continuum limit of this, we embed the $N$ spatial points into a line of length $L$ (with periodic boundary conditions), with lattice spacing $\varepsilon=L/N$.  We must also let the number of timesteps $\tau$ grow as $\varepsilon\rightarrow 0$, so let $t=\tau \varepsilon$, with $t$ constant.  To get a sensible continuum limit, we set $M=m\varepsilon$, where $m$ is a constant.  Defining $p=\mathbf{p}/\varepsilon$ and $\psi_{p}=\psi_{\mathbf{p}}$, and using Trotter's formula \cite{Trot59}, we get
\begin{equation}
\label{eq:Trotter}
\lim_{\varepsilon \to 0}U^{t/\varepsilon}=\exp[-i\displaystyle\sum_{p}\psi^{\dagger}_{p}(p\alpha_1+ m\beta)\psi_{p}t],
\end{equation}
where the sum is now over $p=2\pi k/L$, with $k\in\mathbb{Z}$.  So the continuum limit corresponds to particles evolving via the Dirac Hamiltonian in one spatial dimension:
\begin{equation}
\label{eq:HD}
 H_D=\displaystyle\sum_{p}\psi^{\dagger}_{p}(p\alpha_1+ m\beta)\psi_{p}.
\end{equation}
Furthermore, we could take $L$ to infinity to recover Dirac fermions on an infinite line.

Now that we have found the continuum limit, let us return to the discrete time evolution.  The conditional shift part of the evolution (the unitary $T$ in equation \ref{eq:3}) shifts $\psi_{n,l}$ to the left and $\psi_{n,r}$ to the right.  But we can think of this as one system of $\psi_{n,l}$ fermions evolving via a shift to the left and a copy of that system, the $\psi_{n,r}$ fermions, evolving via the inverse unitary: a shift to the right.  This allows us to apply theorem \ref{th:1} to see that this is equivalent to the fermionic swaps
\begin{equation}
\label{eq:1}
  \psi_{n,l}  \leftrightarrow \psi_{n-1,r},
\end{equation}
at each $n$, followed by
\begin{equation}
\label{eq:2}
   \psi_{n,r}  \leftrightarrow \psi_{n,l},
\end{equation}
at each $n$.  Applying the local unitaries in (\ref{eq:1}) followed by the local unitaries in (\ref{eq:2}) reproduces the original conditional shift in (\ref{eq:3}).

Note also that the part of the evolution that models mass is also a product of local unitaries, since the unitary $W$ in position space is
\begin{equation}
\label{eq:90}
\exp(-iM\displaystyle\sum_{n}\psi^{\dagger}_n\beta\psi_n)=\prod_{n}\exp(-iM\psi^{\dagger}_n\beta\psi_n).
\end{equation}
So the evolution operator $U$ is a product of local unitaries.

In the next section, we represent this discrete fermionic system on qubits.

\subsection{Representation by Qubits}
\label{sec:Representation by Qubits 3}
We associate a qubit to each mode, such that the basis states $\ket{1}_{nl}$ and $\ket{0}_{nl}$ correspond to the presence and absence of a left handed fermion at the point $n$, and similarly the states $\ket{1}_{nr}$ and $\ket{0}_{nr}$ correspond to the presence and absence of a right handed fermion at the point $n$.

We then represent the fermion creation operators using the Jordan-Wigner Transformation, with the ordering $\pi(n,l)=2n$ and $\pi(n,r)=2n+1$, so that
\begin{equation}
 \psi^{\dagger}_{n,\mu} \equiv \spinup_{n\mu}\hspace{-1.2em} \prod_{\pi(k,\nu)<\pi(n,\mu)} \hspace{-1.2em} Z_{k\nu}.
\end{equation}
With this choice of ordering, each local fermionic unitary we found in the last section is local in the qubit picture with the exception of the swap  $\psi^{\dagger}_{0,l} \leftrightarrow \psi^{\dagger}_{N-1,r}$ across the periodic boundary.  For example, the swapping operator in (\ref{eq:1}) for $n \neq0$  becomes
\begin{equation}
\exp[i\frac{\pi}{2}(\spinup_{(n-1)r}-\spinup_{nl})(\spindown_{(n-1)r}-\spindown_{nl})],
\end{equation}
which is a local unitary on the qubits.  The non-local swap crossing the periodic boundary can be dealt with be adding an extra pair of fermionic modes at positions $0$ and $N-1$ and using the trick described in section \ref{sec:More than one spatial dimension}.  Finally, as $W$ is a product of on-site unitaries (equation \ref{eq:90}), these will be local in the qubit picture.  This form of the evolution of discrete Dirac fermions on a (infinite) line was given in \cite{D'Ariano12a,D'Ariano12b,BDT12}.

\subsection{Discrete Dirac Fermions in Three Dimensions}
\label{sec:Discrete Dirac Fermions in Three Dimensions}
In this section, we construct a causal discrete spacetime model that, in the continuum limit, becomes a system of fermions obeying the Dirac equation in three spatial dimensions.  To do this, we first construct discrete fermions that obey the Weyl equation in the continuum limit, and then extend this to fermions obeying the Dirac equation in the continuum limit.  Note that this evolution for a single particle was given in \cite{Bial94}.

Suppose that our discrete space is finite with periodic boundary conditions and sites labelled by three component vectors $\vec{n}$ with components in $\{0,...,N-1\}$.  Suppose also that there are two fermionic modes at each site, labelled by $a=1,2$.  Let us define the corresponding fermion creation operators by $\psi^{\dagger}_{\vec{n},a}$.  Also define
\begin{equation}
 \psi_{\vec{n}}=\begin{pmatrix}
           \psi_{\vec{n},1}\\
           \psi_{\vec{n},2}
          \end{pmatrix}.
\end{equation}
We take $\sigma_i$ to be the Pauli matrices:
\begin{equation}
 \sigma_1=\begin{pmatrix}
            0 & 1\\
            1 & 0\\
           \end{pmatrix},\ \ \sigma_2=\begin{pmatrix}
            0 & -i\\
            i & 0\\
           \end{pmatrix}\ \ \sigma_3=\begin{pmatrix}
            1 & 0\\
            0 & -1\\
           \end{pmatrix}.
\end{equation}

Again, it is simpler to work in momentum space: $\vec{\mathbf{p}}=2\pi \vec{k}/N$ is the discrete momentum, where the components of $\vec{k}$ take values in $\{-\textstyle\frac{N-1}{2},...,\textstyle\frac{N-1}{2}\}$, and we take $N$ to be odd.  The momentum creation operators are
\begin{equation}
\psi^{\dagger}_{\vec{\mathbf{p}},a} =\frac{1}{\sqrt{N^3}}\displaystyle\sum_{\vec{n}}e^{i\vec{\mathbf{p}}.\vec{n}}\psi^{\dagger}_{\vec{n},a},
\end{equation}
with
\begin{equation}
 \psi_{\vec{\mathbf{p}}}=\begin{pmatrix}
           \psi_{\vec{\mathbf{p}},1}\\
           \psi_{\vec{\mathbf{p}},2}
          \end{pmatrix}.
\end{equation}

Suppose that over each timestep the system evolves via the unitary $U=T_1T_2T_3$, where $T_i$ are conditional shifts in each spatial direction:
\begin{equation}
T_i=\exp(-i\displaystyle\sum_{\mathbf{p}}\mathbf{p}_i\psi_{\vec{\mathbf{p}}}^{\dagger}\sigma_i\psi_{\vec{\mathbf{p}}}).
\end{equation}
Note that each $T_i$ is causal.  Analogously to equation \ref{eq:Trotter} in section \ref{sec:Discrete Dirac Fermions in One Dimension}, by applying the Trotter formula, these fermions obey the Weyl equation in the continuum limit.  In other words, in the continuum limit they evolve via the Hamiltonian
\begin{equation}
 H_W=\displaystyle\sum_{\vec{p}}\psi_{\vec{p}}^{\dagger}\vec{p}.\vec{\sigma}\psi_{\vec{p}},
\end{equation}
with the sum ranging over all $\vec{p}=2\pi \vec{k}/L$, where $\vec{k}$ has integer components.

As in the one dimensional case (section \ref{sec:Discrete Dirac Fermions in One Dimension}), we can view the discrete evolution operator $U=T_1T_2T_3$ as a product of local unitaries.  This is because each $T_i$ is a conditional shift operator in the $i$th spatial direction that can be rewritten as a product of local swap operations.  We elaborate on this in appendix \ref{sec:Local Unitaries for Discrete Weyl Fermions in Three Dimensions}.

Now, to get fermions obeying the Dirac equation in the continuum limit, we need {\it four} fermionic modes at each site.  Call the creation operators for these modes $\psi^{\dagger}_{\vec{n},r,a}$ and $\psi^{\dagger}_{\vec{n},l,a}$, with $a=1,2$.  In the continuum limit $r$ and $l$ will correspond to right handed and left handed modes respectively.  We define
\begin{equation}
 \psi_{\vec{\mathbf{p}}}=\begin{pmatrix}
         \psi_{\vec{\mathbf{p}},r}\\
         \psi_{\vec{\mathbf{p}},l}
        \end{pmatrix},
\end{equation}
where each of the components of this vector has two components:
\begin{equation}
 \psi_{\vec{\mathbf{p}},l}=\begin{pmatrix}
         \psi_{\vec{\mathbf{p}},l,1}\\
         \psi_{\vec{\mathbf{p}},l,2}
        \end{pmatrix}\ \textrm{and}\ \psi_{\vec{\mathbf{p}},r}=\begin{pmatrix}
         \psi_{\vec{\mathbf{p}},r,1}\\
         \psi_{\vec{\mathbf{p}},r,2}
        \end{pmatrix}.
\end{equation}

Let the evolution operator be $U=WT$, where in the continuum limit $W$ will contribute the mass term, and $T$ will contribute the momentum term in the Hamiltonian.  Let $T=T_{1}T_{2}T_{3}$, where $T_i$ are conditional shifts acting differently on the $l$ and $r$ fermions:
\begin{equation}
\begin{split}
 T_{i}=&\exp(+i\displaystyle\sum_{\mathbf{p}}\mathbf{p}_i\psi_{\vec{\mathbf{p}},l}^{\dagger}\sigma_i\psi_{\vec{\mathbf{p}},l})\\
\times &\exp(-i\displaystyle\sum_{\mathbf{p}}\mathbf{p}_i\psi_{\vec{\mathbf{p}},r}^{\dagger}\sigma_i\psi_{\vec{\mathbf{p}},r}).
\end{split}
\end{equation}
Note that both terms commute.  Define
\begin{equation}
\beta=\begin{pmatrix}
           0 & I\\
           I & 0
          \end{pmatrix},
\end{equation}
where $I$ is the $2\times 2$ identity matrix.  Also define 
\begin{equation}
\alpha_i=\begin{pmatrix}
           \sigma_i & 0\\
           0 & -\sigma_i
          \end{pmatrix}.
\end{equation}
Then $T_i$ can be rewritten as
\begin{equation}
 T_i=\exp(-i\displaystyle\sum_{\mathbf{p}}\mathbf{p}_i\psi^{\dagger}_{\vec{\mathbf{p}}}\alpha_i\psi_{\vec{\mathbf{p}}}).
\end{equation}

And, similarly to the one dimensional case, we have the mass term
\begin{equation}
 \exp(-iM\displaystyle\sum_{\vec{\mathbf{p}}}\psi^{\dagger}_{\vec{\mathbf{p}}}\beta\psi_{\vec{\mathbf{p}}}).
\end{equation}

As in the one dimensional case, in the continuum limit we get fermions evolving via the three dimensional Dirac Hamiltonian:
\begin{equation}
\label{eq:3DDirac}
 \displaystyle\sum_{\vec{p}}\psi^{\dagger}_{\vec{p}}(\vec{p}.\vec{\alpha}+ m\beta)\psi_{\vec{p}},
\end{equation}
with the sum ranging over all $\vec{p}=2\pi \vec{k}/L$, where $\vec{k}$ has integer components.

The unitary determining the evolution of these discrete Dirac fermions, $U=WT_1T_2T_3$, is equivalent to a product of local unitaries.  This is because $W$ is a product of on-site unitaries mixing between $\psi_{\vec{n},r}$ and $\psi_{\vec{n},l}$ and each $T_i$ can be decomposed separately into local unitaries for both $r$ and $l$ modes (as we saw for the Weyl case above).  To view this as a causal evolution of qubits, however, additional fermionic modes (and hence additional qubits) would have to be introduced to simulate the effects of fermionic anticommutation.  Nevertheless, we saw how to do this in section \ref{sec:More than one spatial dimension}.

\subsection{Fermionic Fields}
\label{sec:Fermionic Fields}
Now let us turn to Quantum Field Theory (QFT), where in the usual approach fields are fundamental, and particles emerge after quantization.  Take the Dirac field, with field operators $\psi_{\alpha}(\vec{x})$, where in three dimensional space $\alpha\in\{1,2,3,4\}$, but in one or two dimensional space $\alpha\in\{1,2\}$.  The field operators obey the Schr\"{o}dinger picture anticommutation relations \cite{Peskin95}
\begin{equation}
\label{eq:13}
\begin{split}
 \{\psi_{\alpha}(\vec{x}),\psi_{\beta}(\vec{y})\} & =0\\
 \{\psi_{\alpha}(\vec{x}),\psi^{\dagger}_{\beta}(\vec{y})\} & =\delta_{\alpha\beta}\delta^{(3)}(\vec{x}-\vec{y}).
\end{split}
\end{equation}

Ideally, we would like to represent individual electrons and positrons at site $\vec{x}$ by fermionic modes at $\vec{x}$, but it is not clear how to do this.  To see why, note that the continuum field operator in three dimensional space is
\begin{equation}
\label{eq:15}
\begin{split}
 & \psi_\alpha (\vec{x})=\\ & \int\!\frac{d^3p}{(2\pi)^3}\frac{1}{\sqrt{2E_p}}\sum_{s=1,2}(a^{s}_{\vec{p}}u^{s}_\alpha (\vec{p})e^{i\vec{p}.\vec{x}}+b^{s\dagger}_{\vec{p}}v^{s}_\alpha (\vec{p})e^{-i\vec{p}.\vec{x}}),
\end{split}
\end{equation}
where $E_p=+\sqrt{|\vec{p}|^2+m^2}$, $a^{s\dagger}_{\vec{p}}$ and $b^{s\dagger}_{\vec{p}}$ create electrons and positrons with momentum $\vec{p}$ and spin state labelled by $s\in\{1,2\}$, and $u_{\alpha}^{s}(\vec{p})$ and $v_{\alpha}^{s}(\vec{p})$ are four component eigenvectors of the Dirac Hamiltonian, satisfying
\begin{equation}
\begin{split}
 & \sum_{\alpha} u^{s\dagger}_\alpha (\vec{p})u^{r}_\alpha (\vec{p})=2E_p\delta_{sr}\\
 &\sum_{\alpha} v_\alpha^{s\dagger}(\vec{p})v_\alpha^{r}(\vec{p})=2E_p\delta_{sr}\\
 &\sum_{\alpha} v^{s\dagger}_\alpha(-\vec{p})u^{r}_\alpha(\vec{p})=0.
\end{split}
\end{equation}
See \cite{Peskin95} for more details.  Now the field operator in (\ref{eq:15}) obeys the anticommutation relations (\ref{eq:13}), but the electron part on its own, given by
\begin{equation}
A_{\alpha}(\vec{x}) =\int \frac{d^3p}{(2\pi)^3}\frac{1}{\sqrt{2E_p}}\sum_{s=1,2}a^{s}_{\vec{p}}u^{s}_{\alpha}(\vec{p})e^{i\vec{p}.\vec{x}},
\end{equation}
does not satisfy
\begin{equation}
 \{A_{\alpha}(\vec{x}),A^{\dagger}_{\beta}(\vec{y})\}=\delta_{\alpha\beta}\delta^{(3)}(\vec{x}-\vec{y}).
\end{equation}
This means we cannot assign separate spatial fermionic modes to represent both an electron field and positron field.  This is related to the problem of not being able to localize electrons with only positive energy wavefunctions.  So we will work with $\psi_{\alpha}(\vec{x})$ directly.

In the continuum, the field operator $\psi_{\alpha}(\vec{x})$ obeys the Dirac equation.  For a discrete model, we represent the field by fermionic modes at each point evolving via the discrete spacetime evolution that we considered in section \ref{sec:Discrete Dirac Fermions in One Dimension} or \ref{sec:Discrete Dirac Fermions in Three Dimensions}.  We know that the continuum limit of this is free fermions obeying the Dirac equation.  This is essentially what is presented in \cite{D'Ariano12b}.  Each particle evolves in accordance with the Dirac equation in the continuum limit, but they do not always correspond to particles with positive energy.  This is essentially because the vacuum is taken to be the state annihilated by $\psi_{\vec{n},a}$ for all $\vec{n}$ and $a$.

The naive vacuum defined by $\psi_{\alpha}(\vec{x})\ket{\Omega}=0$ for all $\vec{x}$ and  $\alpha$ does not correspond to the physical vacuum in QFT, as the physical vacuum has no electrons or positrons present.  Hence, we require that the physical vacuum is annihilated by all $b^s_{\vec{p}}$ and $a^s_{\vec{p}}$.  We can create this physical vacuum $\ket{\Omega_D}$ by acting on $\ket{\Omega}$ with all $b^s_{\vec{p}}$ operators, which ensures that $b^s_{\vec{p}}\ket{\Omega_D}=0$.  (This is equivalent to the Dirac sea picture.  Viewed in that way, $b^s_{\vec{p}}$ creates a negative energy electron, and $b^{s\dagger}_{\vec{p}}$ creates a hole in the sea of negative energy particles.)  So in the discrete case we need to consider a new vacuum state analogous to the Dirac sea state in the continuum.  We postpone a detailed discussion of this problem for future work.

Finally, note that in \cite{DdV87} a causal fermionic model in discrete spacetime is given that becomes the massive Thirring model in one spatial dimension in the continuum limit.  In the massive Thirring model, evolution is governed by the Hamiltonian
\begin{equation}
 H=H_D+\frac{1}{2}g\!\int\! \textrm{d}xj^{\mu}(x)j_{\mu}(x),
\end{equation}
where $g$ is a constant, $H_D$ is the free continuum Dirac Hamiltonian in one spatial dimension (equation \ref{eq:HD} in section \ref{sec:Discrete Dirac Fermions in One Dimension}) and $j^{\mu}(x)$ is the current: $j^0(x)=\psi^{\dagger}(x)\psi(x)$ and $j^i(x)=\psi^{\dagger}(x)\alpha^i\psi(x)$.

\section{Simulating Causal Fermions}
\label{sec:Simulating Causal Fermions}
We have seen how to represent a causal fermionic unitary by applying local unitaries to a lattice of qubits, which is a type of quantum cellular automaton.  This tells us how to simulate the evolution of causal fermions on a quantum computer.

From a complexity point of view, simulating the evolution of these systems can be done efficiently because, for $\mathcal{N}$ fermionic modes, we have to apply $O(\mathcal{N})$ local unitaries: first the $O(\mathcal{N})$ unitaries $V_{\vec{x}\mu}$ (the unitaries implementing $U_B\mathbf{S}_{\vec{x}\mu}U_B^{\dagger}$ in the qubit representation), followed by the qubit realisation of the $O(\mathcal{N})$ fermionic swap operators.  As we saw, we may need to include additional qubits to ensure that these operators are local in the qubit picture.  Note also that the cost of applying $V_{\vec{x}\mu}$ does not grow with $\mathcal{N}$ \footnote{We take it for granted that the region on which $U_B\mathbf{S}_{\vec{x}\mu}U_B^{\dagger}$ is localized, the neighbourhood of $\vec{x}$, does not grow with $\mathcal{N}$.  This is a natural requirement: the discrete Dirac fermions we study in section \ref{sec:Discrete Dirac Fermions in One Dimension} have this property.}.  Furthermore, a lot of these operations can be done in parallel: all swap operations can be done simultaneously, while we can do many $V_{\vec{x}\mu}$ operations at the same time, provided the areas on which these unitaries are localized do not overlap.  As each $V_{\vec{x}\mu}$ is localized on hypercubes with length of side $3$ (because the evolution is causal), the time needed to implement one step of the evolution is $O(3^d)$, where $d$ is the lattice dimension, so the time does not depend on $\mathcal{N}$.

So much for implementing the evolution on a quantum computer.  We still have to prepare a $+1$ eigenstate of the pairs of Majorana fermions on the quantum computer.  This can be done efficiently by using a method presented in \cite{AL97} to deal with the strings of $Z$ operators that arise in the qubit representation.  The method is given in appendix \ref{sec:Preparing the Majorana state}.

Finally, there additional subtleties involved when trying to simulate quantum field theories.  These are discussed in \cite{JLP12}, which shows how to simulate $\phi^4$ theory on a quantum computer in a very different manner.

\section{Discussion}
\label{sec:Discussion}
Throughout this paper, we have discussed causal quantum systems in discrete spacetime.  The first results we obtained were theorem \ref{th:1} and corollary \ref{cor:1}, which allow us to decompose causal unitaries into a product of local unitaries in a manner analogous to that of \cite{ANW11,GNVW12}.  Later, we used this and a method for mapping local fermionic Hamiltonians to local qubit Hamiltonians to prove theorem \ref{th:2}, which showed that causal fermionic quantum systems in discrete spacetime can be viewed as lattices of qubits evolving causally in discrete time, meaning they are types of quantum cellular automaton.  After discussing specific discrete spacetime fermionic models, we showed why these systems can be efficiently simulated on a quantum computer.

The next objective is to devise causal discretized models that become interesting quantum field theories in the continuum limit, such as quantum electrodynamics or quantum chromodynamics, with the ultimate goal being the construction of causal models that reproduce the entire standard model in the continuum limit.  One of the reasons this is interesting is that this would allow simulation of these models by a quantum computer, something that seems feasible for purely fermionic systems.  For bosonic systems any such program would require a cut-off to allow the state of each bosonic mode to be represented by a finite number of qubits.

We hope that the results contained here may help in some way to find causal discrete spacetime models that converge to interesting physical systems in the continuum limit.  This would not only be useful for simulation but also as mathematically sensible (discretized) quantum field theories with a strict notion of causality that offer an alternative view to current discretized models.  Furthermore, the study of causal discrete models may even hint at physics beyond the standard model, particularly as it is sometimes suggested that at some small scale the notion of continuous spacetime may break down.

\acknowledgments
AJS acknowledges support from the Royal Society.  TCF acknowledges support from the Robert Gardiner Memorial Scholarship, CQIF, DAMTP, EPSRC and St John's College, Cambridge.  We would also like to thank Fay Dowker for pointing out reference \cite{DdV87}.

\bibliographystyle{unsrt}
\bibliography{QCA1,QFT1,Walks1,Other1,Sim1}

\appendix
\section{Continuous Time Models in Discrete Space}
\label{sec:Continuous Time Models in Discrete Space}
Here we will see why we cannot construct interesting discrete space systems in continuous time that are strictly causal.  Take a particle on a finite discrete line evolving in continuous time with some time independent Hamiltonian $H$.  Our strict notion of causality would imply that, if the particle is at position $0$ at $t=0$, then there is some $T$ such that for $t<T$ the particle has zero probability of being found outside a finite region $R$ containing $0$.

This tells us that, given any position $n$ outside $R$, we must have that, for all $t<T$, $\bra{n}e^{-iHt}\ket{0}=0$.  Expanding $e^{-iHt}$,
\begin{equation}
\begin{split}
 \bra{n}e^{-iHt}\ket{0} & =\bra{n}-iHt+O(t^2)\ket{0}=0\\
 \Rightarrow\ &\bra{n}iH\ket{0}-\bra{n}O(t)\ket{0}=0,
\end{split}
\end{equation}
but the second term can be made arbitrarily small by taking $t$ to be small, so the first term must be zero.  Similarly, by looking at higher order terms in the expansion of $e^{-iHt}$, we see that $\bra{n}H^l\ket{0}=0$ for any $l$.  But this implies that $\bra{n}e^{-iHt}\ket{0}=0$ for \emph{any} $t$.  It follows from this that the particle will not be found outside $R$ at any $t$.

Therefore, if the particle is ever going to propagate to a point $n$, it happens instantaneously, though possibly with a very small amplitude.

\section{Proofs of Lemmas \ref{lem:1} and \ref{lem:2}}
\label{sec:Proofs of Lemmas 1 and 2}
Here we prove lemmas \ref{lem:1} and \ref{lem:2}.  It will be useful to repeat the extra requirement we made of causal unitaries in section \ref{sec:Causal Fermions}.  Given a system of fermions evolving via $U$, we can add additional fermionic modes whose creation operators anticommute with the original fermion creation and annihilation operators while commuting with $U$.  This is essentially a discrete time analogue of the requirement that Hamiltonians are sums of even products of creation and annihilation operators in continuous time systems.

For simplicity of notation, we will assume here that these fermions have no extra degrees of freedom, though the extension to systems of fermions with extra degrees of freedom is straighforward.
\begin{lemma}
\label{lem:2}
 Given a unitary $U$ acting on fermions with annihilation operators $a_{\vec{x}}$, $U^{\dagger}a_{\vec{x}}U$ is a linear combination of odd products of fermion creation and annihilation operators.
\end{lemma}
\begin{proof}
We write $U^{\dagger}a_{\vec{x}}U=A_{odd}+A_{even}$, where $A_{odd}$ are all the terms that are products of an odd number of creation and annihilation operators and $A_{even}$ are all terms that are products of an even number of creation and annihilation operators.

Our extra requirement above implies that we can add a fermionic mode with creation operator $b^{\dagger}$, which anticommutes with all of the original creation and annihilation operators while commuting with $U$.  But this implies that $U^{\dagger}\{b^{\dagger},a_{\vec{x}}\}U=\{b^{\dagger},A_{odd}+A_{even}\}=0$, which is only possible if $A_{even}=0$.
\end{proof}
Next we prove Lemma \ref{lem:1} from section \ref{sec:Causal Fermions}.
\newtheorem*{lemma:1}{Lemma \ref{lem:1}}
\begin{lemma:1}
The inverse of a causal fermionic unitary $U$ is also a causal fermionic unitary.
\end{lemma:1}
\begin{proof}
Lemma \ref{lem:2} tells us that $U^{\dagger}a_{\vec{y}}U$ must be a linear combination of odd products of creation and annihilation operators.  So, since U is causal,
\begin{equation}
 \begin{split}
&\{U^{\dagger}a_{\vec{y}}U,a_{\vec{x}}\} =0,\\
\textrm{and}\ & \{U^{\dagger}a^{\dagger}_{\vec{y}}U,a_{\vec{x}}\}= 0
\end{split}
\end{equation}
for all $\vec{y}$ when $\vec{x}$ is not in the neighbourhood of $\vec{y}$.  Then
\begin{equation}
\begin{split}
&\{a_{\vec{y}},Ua_{\vec{x}}U^{\dagger}\}= 0,\\
\textrm{and}\ & \{a^{\dagger}_{\vec{y}},Ua_{\vec{x}}U^{\dagger}\}= 0
\end{split}
\end{equation}
for all $\vec{y}$ when $\vec{x}$ is not in the neighbourhood of $\vec{y}$.  But $\vec{x}$ is in the neighbourhood of $\vec{y}$ implies $\vec{y}$ is in the neighbourhood of $\vec{x}$.  It follows that $Ua_{\vec{x}}U^{\dagger}$ is localized on the neighbourhood of $\vec{x}$, so $U^{\dagger}$ is causal.
\end{proof}

\section{The Fermionic swap operator}
\label{app:swap} 
Given two fermionic modes with annihilation operators $a$ and $b$, it is useful to define a unitary operator which swaps them:
\begin{equation} 
S^{\dagger} a S = b, \qquad S^{\dagger} b S = a.
\end{equation}
Here we show that this unitary is given by 
\begin{equation} 
S= \exp[i\frac{\pi}{2}(b^{\dagger}-a^{\dagger})(b-a)],
\end{equation} 
which is self-adjoint.  To see this, define two new fermionic modes 
\begin{equation}
c = \frac{1}{\sqrt{2}} (b - a),  \qquad d =  \frac{1}{\sqrt{2}} (b + a),
\end{equation} 
which satisfy the usual anticommutation relations. Then $S =  \exp[i \pi c^{\dagger} c ]$, and 
\begin{equation}
S a S =e^{i \pi c^{\dagger} c } \frac{1}{\sqrt{2}} \left(d -c   \right)e^{-i \pi c^{\dagger} c}  \nonumber \\
=  \frac{1}{\sqrt{2}}(d + c)  = b,
\end{equation} 
where we have used the fact that $(c^{\dagger} c) c= 0$, and $c (c^{\dagger} c) = c$. Similarly, it is easy to see that $S b S = a$.

\section{Representing Local Fermionic Hamiltonians by Local Qubit Hamiltonians}
\label{sec:Representing Local Fermionic Hamiltonians by Local Qubit Hamiltonians}
Here we will see the main idea from \cite{Ball05,VC05}.  In section \ref{sec:More than one spatial dimension} and appendix \ref{sec:Representation by Qubits 2}, we apply this technique to local fermionic operators.  In \cite{Ball05,VC05}, however, they apply it to local fermionic Hamiltonians, which are sums of local fermionic operators.

To illustrate the idea, let us look at a local fermionic Hamiltonian.  By including additional fermionic modes, we will construct a local fermionic Hamiltonian that replicates the dynamics of the original Hamiltonian but is local in the qubit representation.  As an example, take
\begin{equation}
\label{eq:5}
 H=\sum_{<\vec{x}\vec{y}>}(a^{\dagger}_{\vec{x}}a^{\ }_{\vec{y}}+a^{\dagger}_{\vec{y}}a^{\ }_{\vec{x}}),
\end{equation}
where $\langle \vec{x}\vec{y}\rangle$ denotes nearest neighbour pairs, and we are considering a rectangular lattice as shown in figure \ref{fig:1}.
\begin{figure}[!ht]
{\centering
\resizebox{5.5cm}{!}{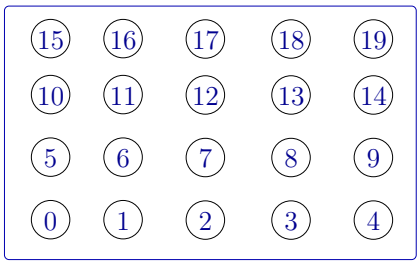} \caption{A lattice of $5 \times 4$ different sites, with Jordan-Wigner Transformation ordering, $\pi(\vec{x})$, as shown.}\label{fig:1}
}
\end{figure}
With the Jordan-Wigner Transform given by (\ref{eq:general ordering}), and the ordering for the Jordan-Wigner Transformation shown in figure \ref{fig:1}, all of the vertical hopping terms in (\ref{eq:5}) will be non local in the spin picture.  Furthermore, with any other choice of ordering, some terms in (\ref{eq:5}) would be non local in the spin picture.  And, as we consider bigger and bigger lattices, the length of the strings of $Z$ operators in the non local terms will grow.

We want to eliminate these non local strings of $Z$s.  To do this, we will introduce extra fermions to cancel the strings of $Z$s in the qubit representation.

In particular, we introduce a pair of additional fermionic modes whenever two sites are connected in the Hamiltonian by a hopping term $a_{\vec{x}}^{\dagger} a_{\vec{y}} + a_{\vec{y}}^{\dagger} a_{\vec{x}}$ that is not local in the qubit picture (for example, when $\pi(\vec{x})=0$, $\pi({\vec{y}})=5$ in figure \ref{fig:1}).  We introduce one additional fermionic mode at site ${\vec{x}}$ and one at site ${\vec{y}}$, with creation operators denoted by $c_{(\vec{x},\vec{y})}^{\dagger}$ and $c_{(\vec{y},\vec{x})}^{\dagger}$ respectively. The first index in the subscript gives the location of the ancillary fermion, and the second indicates the hopping destination.  We define the operators
\begin{equation}
\begin{split}
 m_{(\vec{x},\vec{y})} & =c_{(\vec{x},\vec{y})}+c_{(\vec{x},\vec{y})}^{\dagger}\\
m_{(\vec{y},\vec{x})} & =c_{(\vec{y},\vec{x})}+c_{(\vec{y},\vec{x})}^{\dagger},
\end{split}
\end{equation}
which are self-adjoint and satisfy
\begin{equation}
 \begin{split}
\{m_{(\vec{x},\vec{y})},m_{(\vec{w},\vec{z})}\} & =2\delta_{\vec{x}\vec{w}}\delta_{\vec{y}\vec{z}}\\
\{m_{(\vec{x},\vec{y})},a_{\vec{z}}\} & =0.
\end{split}
\end{equation}
These can be thought of as Majorana fermion operators.

Now define the operator $M_{(\vec{x},\vec{y})}=im_{(\vec{x},\vec{y})}m_{(\vec{y},\vec{x})}$ and note that it is self-adjoint and has eigenvalues $+1$ and $-1$ since $M_{(\vec{x},\vec{y})}^2=1$ and $M_{(\vec{x},\vec{y})}\neq -1$.

Next, we make the substitution 
\begin{equation} 
a_{\vec{x}}^{\dagger} a_{\vec{y}} + a_{\vec{y}}^{\dagger} a_{\vec{x}}\rightarrow a_{\vec{x}}^{\dagger} M_{(\vec{x},\vec{y})} a_{\vec{y}} + a_{\vec{y}}^{\dagger} M_{(\vec{x},\vec{y})} a_{\vec{x}}
\end{equation}
in the Hamiltonian in \ref{eq:5} whenever sites $\vec{x}$ and $\vec{y}$ are not adjacent in the ordering scheme. As all of the operators $M_{(\vec{x},\vec{y})}$ commute, there exists a joint eigenstate with eigenvalue $+1$ for each operator.  (This is the reason we must distinguish between the ancillary modes $c_{(\vec{x},\vec{y})}$ and $c_{(\vec{x},\vec{z})}$. If we were to replace both with a single mode $c_{\vec{x}}$  then the operators $M_{(\vec{x},\vec{y})}$ and $M_{(\vec{x},\vec{z})}$ would not commute.)  When the ancillary modes  are prepared in this state, then the action of the transformed Hamiltonian on the original fermions will be the same as that of the original Hamiltonian.

It is natural to choose the new ordering such that fermions at the same site are consecutive because, as we mentioned in section \ref{sec:Fermions and the Jordan-Wigner Transformation}, this means that the product of an even number of fermionic creation and annihilation operators at the same site will always be local in the qubit representation.  Then, because $m_{(\vec{x},\vec{y})}$ is a fermionic operator at site $\vec{x}$, the operator $a^{\dagger}_{\vec{x}}m_{(\vec{x},\vec{y})}$ is local in the qubit representation.  Therefore it follows that
\begin{equation} 
a_{\vec{x}}^{\dagger} M_{(\vec{x},\vec{y})} a_{\vec{y}} + a_{\vec{y}}^{\dagger} M_{(\vec{x},\vec{y})} a_{\vec{x}}
\end{equation}
is local in the qubit representation.  The same trick will work for any local fermionic Hamiltonian that is quadratic in fermion creation and annihilation operators.

Furthermore, even if the fermionic Hamiltonian has higher order terms (provided they only contain even products of creation and annihilation operators), like $a^{\dagger}_{\vec{x}}a^{\dagger}_{\vec{y}}a_{\vec{w}}a_{\vec{z}}+a^{\dagger}_{\vec{z}}a^{\dagger}_{\vec{w}}a_{\vec{y}}a_{\vec{x}}$, for example, we can still make this local in the qubit picture by adding more Majorana fermions.  For this example, we can replace $a^{\dagger}_{\vec{x}}a^{\dagger}_{\vec{y}}a_{\vec{w}}a_{\vec{z}}$ by
\begin{equation}
a^{\dagger}_{\vec{x}}M_{(\vec{x},\vec{y})}a^{\dagger}_{\vec{y}}a_{\vec{w}}M_{(\vec{w},\vec{z})}a_{\vec{z}}.
\end{equation}

By using this trick, any local fermionic Hamiltonian has a corresponding local qubit Hamiltonian.  Also, each term in the qubit Hamiltonian connects the same sets of lattice sites as the corresponding term in the original fermionic Hamiltonian 

Furthermore, because $[a^{\dagger}_{\vec{z}},M_{(\vec{x},\vec{y})}]=0$ for any $\vec{x}$, $\vec{y}$ and $\vec{z}$, it follows that we can act on a $+1$ eigenstate of all the $M_{(\vec{x},\vec{y})}$ terms with physical fermion creation operators $a^{\dagger}_{\vec{z}}$ and still have a $+1$ eigenstate.  Therefore the original fermionic dynamics can be viewed as a subsector of the dynamics of the corresponding local qubit Hamiltonian.

Note that the requirement that the Hamiltonian is a sum of even products of creation and annihilation operators was necessary to use this trick.  Furthermore, the fact that we could preserve the self-adjointness of the Hamiltonian was also crucial.

Finally, note that the number of additional Majorana fermions we need does not depend on the number of physical fermionic modes at each site: we only need one pair of Majorana fermions for all terms in the Hamiltonian connecting a particular pair of sites that are non local in the qubit representation.

\section{Local Unitaries for Discrete Weyl Fermions in Three Dimensions}
\label{sec:Local Unitaries for Discrete Weyl Fermions in Three Dimensions}
Here we will see that we can decompose the discrete evolution operator $U=T_1T_2T_3$ into a product of local unitaries.  Each $T_i$ is a conditional shift operator in the $i$th spatial direction that can be rewritten as a product of local swap operations.

It will be convenient to change notation slightly, so that
\begin{equation}
 \psi_{\vec{n}}=\begin{pmatrix}
           \psi_{(\vec{n},\uparrow_z)}\\
           \psi_{(\vec{n},\downarrow_z)}
          \end{pmatrix}.
\end{equation}
Then $T_3$, the conditional shift in the $z$ direction, is
\begin{equation}
\begin{split}
T_3\psi_{(\vec{n},\uparrow_z)}T_3^{\dagger}&=\psi_{(\vec{n}+\vec{e}_z,\uparrow_z)}\\
T_3\psi_{(\vec{n},\downarrow_z)}T_3^{\dagger}&=\psi_{(\vec{n}-\vec{e}_z,\downarrow_z)},
\end{split}
\end{equation}
where $\vec{e}^{\ T}_z=(0,0,1)$.  But we already know how to rewrite this in terms of local swap operations, which we did in the one dimensional case in equations \ref{eq:1} and \ref{eq:2}.  So $T_3$ is equivalent to applying
\begin{equation}
  \psi_{(\vec{n},\downarrow_z)}  \leftrightarrow \psi_{(\vec{n}-\vec{e}_z,\uparrow_z)},
\end{equation}
at each $\vec{n}$, followed by
\begin{equation}
   \psi_{(\vec{n},\uparrow_z)}  \leftrightarrow \psi_{(\vec{n},\downarrow_z)},
\end{equation}
at each $\vec{n}$.  Similarly, $T_1$ and $T_2$ are conditional shifts depending on the internal degree of freedom of the particle.  We define
\begin{equation}
\begin{split}
\psi_{(\vec{n},\uparrow_x)}=\textstyle{\frac{1}{\sqrt{2}}}(\psi_{(\vec{n},\uparrow_z)}+\psi_{(\vec{n},\downarrow_z)})\\
\psi_{(\vec{n},\downarrow_x)}=\textstyle{\frac{1}{\sqrt{2}}}(\psi_{(\vec{n},\uparrow_z)}-\psi_{(\vec{n},\downarrow_z)}),
\end{split}
\end{equation}
so that
\begin{equation}
\begin{split}
T_1\psi_{(\vec{n},\uparrow_x)}T_1^{\dagger}&=\psi_{(\vec{n}+\vec{e}_x,\uparrow_x)}\\
T_1\psi_{(\vec{n},\downarrow_x)}T_1^{\dagger}&=\psi_{(\vec{n}-\vec{e}_x,\downarrow_x)},
\end{split}
\end{equation}
where $\vec{e}^{\ T}_x=(1,0,0)$.  And this has a decomposition in terms of local swaps also:
\begin{equation}
  \psi_{(\vec{n},\downarrow_x)}  \leftrightarrow \psi_{(\vec{n}-\vec{e}_x,\uparrow_x)},
\end{equation}
at each $\vec{n}$, followed by
\begin{equation}
   \psi_{(\vec{n},\uparrow_x)}  \leftrightarrow \psi_{(\vec{n},\downarrow_x)},
\end{equation}
at each $\vec{n}$.  An analogous decomposition holds for $T_2$, with
\begin{equation}
\begin{split}
\psi_{(\vec{n},\uparrow_y)}=\textstyle{\frac{1}{\sqrt{2}}}(\psi_{(\vec{n},\uparrow_z)}-i\psi_{(\vec{n},\downarrow_z)})\\
\psi_{(\vec{n},\downarrow_y)}=\textstyle{\frac{1}{\sqrt{2}}}(\psi_{(\vec{n},\uparrow_z)}+i\psi_{(\vec{n},\downarrow_z)}).
\end{split}
\end{equation}
Note that $\psi^{\dagger}_{(\vec{n},\uparrow_y)}$ creates a fermion with spin up in the $y$ direction.

\section{Infinite Systems}
\label{sec:Infinite Systems}
Here we will prove our main results when the spatial lattice is infinite.  Note that, if we are concerned only with finite times and finite regions of space, which is the case for simulations, we do not need to consider infinite lattices, as the evolution is strictly causal.  That said, the extension to infinite lattices emphasizes the fact that causal fermions are analogues of quantum cellular automata, and that theorem \ref{th:1} is a fermionic analogue of the main result of \cite{ANW11,GNVW12}, as that result is proved for systems on infinite lattices.

When dealing with quantum systems composed of infinitely many subsystems, it is not clear at first glance what we should take as our Hilbert space.  This is because an infinite tensor product of Hilbert spaces is not sensible.  One way to get around this is to take local operators to be fundamental and represent them by elements of an abstract algebra.  Then define states as functionals of the elements of the algebra.  This is the C*-algebra approach \cite{BR97}.

We will now give the precise definition of a C*-algebra.  First, a complex algebra is a complex vector space with a product operation that is associative and distributive over addition.
\begin{defin}
 A C*-algebra $\mathcal{A}$ is a complex algebra with a norm $\|\cdot\|$, in which it is complete, and an anti-linear map $A\rightarrow A^{*}$, with the following properties:
\begin{enumerate}
\centering
 \item $(AB)^*=B^*A^*$
 \item $\|AB\|\leq \|A\|\|B\|$
 \item $\|A^*\|=\|A\|$
 \item $\|A^*A\|=\|A\|^2.$
\end{enumerate}
\end{defin}
A simple example of a C*-algebra is $\mathcal{M}_n(\mathbb{C})$, the set of $n\times n$ complex matrices, where the norm is the spectral norm (the largest singular value) and the * operation is the hermitian conjugate.  Because it is finite dimensional, this example misses out on the subtleties associated with infinite dimensional vector spaces.

We assume that all C*-algebras we consider have an identity, denoted $I$.

Next we define states on the C*-algebra.
\begin{defin}
 A state on a C*-algebra $\mathcal{A}$ is a linear functional $\rho$ that is positive, meaning $\rho(B^*B)\geq 0\ \forall B\in\mathcal{A}$, and normalized, meaning $\rho(I)=1$.
\end{defin}
For an infinite spin chain, a simple example of a state is all spins pointing up in the x-direction.  In the finite dimensional case, for any state $\rho$ there is a density operator $\sigma$ such that $\rho(A)=\textrm{tr}[\sigma A]$ for any $A$ in the C*-algebra.

Next, we require that the evolution is an invertible map on the C*-algebra that preserves its structure.  Such a map is called an automorphism.
\begin{defin}
An automorphism $\alpha$ is an invertible linear map on the C*-algebra that satisfies
\begin{enumerate}
 \item \centering $\alpha(A)\alpha(B)=\alpha(AB)$
 \item $\alpha(A^*)=\alpha(A)^*$
 \item $\|\alpha(A)\|=\|A\|$.
\end{enumerate}
\end{defin}
The first property implies that the dynamics preserve commutation or anticommutation relations.  An example of an automorphism is $A\rightarrow U^{*}AU$ for any $U$ in the C*-algebra satisfying $UU^*=U^*U=I$.

As we are working in a discrete spacetime picture, there is a subalgebra associated to every spatial point.  We can define a notion of causality for these systems that is a natural extension of the definitions we had in the finite case.
\begin{defin}
An automorphism $\alpha$, is causal if, for any $\vec{x}$ and any $A$ localized on $\vec{x}$, $\alpha(A)$ is localized on the neighbourhood of $\vec{x}$.
\end{defin}
So, if an automorphism is causal, then observables on $\vec{x}$ cannot spread by more  than one step in each direction (i.e. the size of the neighbourhood of $\vec{x}$) after every timestep.

\subsection{Quantum Lattice Systems and QCA}
The C*-algebra for a quantum lattice system is defined by associating elements of the algebra to finite regions of the lattice, with the property that elements associated to $\Lambda$ and $\Lambda^{\prime}$ commute if $\Lambda \cap \Lambda^{\prime}=\varnothing$.  Furthermore, the set of all elements associated to a finite $\Lambda$ is isomorphic to the set of operators we get by assigning finite dimensional Hilbert spaces to the systems in $\Lambda$.  Also, the norm of elements in a finite region $\Lambda$ is just the operator norm on the corresponding operators in the Hilbert space picture.  So locally the C*-algebra looks like a finite quantum system.  For the example of a line of qubits, we have that the algebra associated to each site is equivalent to $\mathcal{M}_2(\mathbb{C})$.

It is a useful result \cite{BR97} that to specify a state we need only specify a family of states $\rho_{\Lambda}$ on every finite region $\Lambda$, with the consistency condition that, if $\Lambda\subseteq \Lambda^{\prime}$,
\begin{equation}
 \rho_{\Lambda}(A)=\rho_{\Lambda^{\prime}}(A),
\end{equation}
where $A$ is an element of the algebra associated to region $\Lambda$.

The precise definition of a Quantum Cellular Automaton is as follows.
\begin{defin}
 A Quantum Cellular Automaton is a quantum lattice system together with evolution over discrete timesteps via a causal automorphism.
\end{defin}

\subsection{Fermions}
\label{sec:Fermions}
For fermions, the C*-algebra is generated by objects satisfying the canonical anticommutation relations:
\begin{equation}
 \begin{split}
 \{a^{\dagger}_{\vec{x}},a_{\vec{y}}\} & =\delta_{\vec{x}\vec{y}}\\
 \{a_{\vec{x}},a_{\vec{y}}\} & =0,
\end{split}
\end{equation}
as in the finite case, but now $\vec{x}\in \mathbb{Z}^d$.  We will refer to $a^{\dagger}_{\vec{x}}$ and $a_{\vec{x}}$ as creation and annihilation operators even though they are no longer operators, rather elements of an abstract algebra.  Also, we will use a dagger to denote the * operation.  To simplify notation, we will assume that there is only one fermionic mode at each spatial point, but all of the following results hold with extra degrees of freedom.

As in the finite case, we say that a fermionic operator is localized on a spatial region $R$ if it can be written in terms of creation and annihilation operators corresponding only to $R$.

In section \ref{sec:Causal Fermions}, we justified an extra requirement on the evolution of fermions: given a system of fermions evolving via a unitary $U$, we can add additional fermionic modes whose creation operators anticommute with the original fermion creation operators while commuting with $U$.  Here, we make an analogous requirement on the evolution.
\begin{requirement}
\label{rec:1}
Given a system of fermions evolving via a causal automorphism $\alpha$, we can add additional fermionic modes whose creation operators anticommute with the original fermion creation operators but are unchanged by $\alpha$.
\end{requirement}  

Before moving on to the local decomposition for causal evolutions, we will prove two useful lemmas.
\begin{lemma}
\label{lem:3}
 Given a causal automorphism $\alpha$ of fermions with annihilation operators $a_{\vec{x}}$ satisfying the above requirement, $\alpha(a_{\vec{x}})$ is a linear combination of odd products of fermion creation and annihilation operators.
\end{lemma}
\begin{proof}
We write $\alpha(a_{\vec{x}})=A_{odd}+A_{even}$, where $A_{odd}$ is the sum of all the terms that are products of an odd number of creation and annihilation operators and $A_{even}$ is the sum of all terms that are products of an even number of creation and annihilation operators.

Our extra requirement above implies that we can add a fermionic mode with creation operator $b^{\dagger}$, which anticommutes with all of the original creation and annihilation operators while having $\alpha(b^{\dagger})=b^{\dagger}$.  But this implies that $\alpha(\{b^{\dagger},a_{\vec{x}}\})=\{b^{\dagger},A_{odd}+A_{even}\}=0$, which is only possible if $A_{even}=0$.
\end{proof}
Next we give another useful lemma.
\begin{lemma}
The inverse of a causal fermionic evolution $\alpha$ is also a causal fermionic evolution.
\end{lemma}
\begin{proof}
Lemma \ref{lem:3} tells us that $\alpha(a_{\vec{y}})$ must be a linear combination of odd products of creation and annihilation operators.  So, since $\alpha$ is causal,
\begin{equation}
 \begin{split}
&\{\alpha(a_{\vec{y}}),a_{\vec{x}}\} =0,\\
\textrm{and}\ & \{\alpha(a^{\dagger}_{\vec{y}}),a_{\vec{x}}\}= 0
\end{split}
\end{equation}
for all $\vec{y}$ when $\vec{x}$ is not in the neighbourhood of $\vec{y}$.  Since $\alpha^{-1}$ is an automorphism,
\begin{equation}
\begin{split}
&\{a_{\vec{y}},\alpha^{-1}(a_{\vec{x}})\}= 0,\\
\textrm{and}\ & \{a^{\dagger}_{\vec{y}},\alpha^{-1}(a_{\vec{x}})\}= 0
\end{split}
\end{equation}
for all $\vec{y}$ when $\vec{x}$ is not in the neighbourhood of $\vec{y}$.  But $\vec{x}$ is in the neighbourhood of $\vec{y}$ implies $\vec{y}$ is in the neighbourhood of $\vec{x}$.  It follows that $\alpha^{-1}(a_{\vec{x}})$ is localized on the neighbourhood of $\vec{x}$, so $\alpha^{-1}$ is causal.
\end{proof}

We are now ready to decompose causal evolutions into products of local unitaries.

\subsection{Local Decomposition of Causal Evolutions}
We denote the automorphism that swaps a fermionic mode, with annihilation operator $a_{\vec{x}}$, and its copy, with annihilation operator $b_{\vec{x}}$, by $s_{\vec{x}}$.  This means that $s_{\vec{x}}(a_{\vec{x}})=b_{\vec{x}}$ and $s_{\vec{x}}(b_{\vec{x}})=a_{\vec{x}}$.  Note that we denote the composition of two automorphisms $\alpha_1$ and $\alpha_2$ by $\alpha_1\alpha_2$, meaning $\alpha_1\alpha_2(A)=\alpha_1(\alpha_2(A))$.
\begin{theorem}
Given a system of fermions, evolving via a causal evolution $\alpha$, the evolution of two copies of this system via $\alpha\beta^{-1}$, where $\beta$ is equivalent to $\alpha$ but acting on the copy system, can be decomposed into local unitaries:
\begin{equation}
 \alpha\beta^{-1}=\prod_{\vec{x}}(s_{\vec{x}})\prod_{\vec{y}}[\beta s_{\vec{y}}\beta^{-1}],
\end{equation}
where $s_{\vec{x}}$ and $\beta s_{\vec{y}}\beta^{-1}$ are equivalent to conjugation by commuting local unitaries.
\end{theorem}
\begin{proof}
First,
\begin{equation}
 \prod_{\vec{x}}(s_{\vec{x}})\prod_{\vec{y}}[\beta s_{\vec{y}}\beta^{-1}]=s\beta s\beta^{-1},
\end{equation}
where
\begin{equation}
s=\prod_{\vec{x}}s_{\vec{x}}
\end{equation}
is the global swap.  This implies $s\beta s=\alpha$.  It follows that
\begin{equation}
 \prod_{\vec{x}}(s_{\vec{x}})\prod_{\vec{y}}[\beta s_{\vec{y}}\beta^{-1}]=\alpha\beta^{-1}.
\end{equation}
Furthermore, for fermions, applying $\beta s_{\vec{y}}\beta^{-1}$ is equivalent to conjugation by a local fermionic unitary because $s_{\vec{x}}(A)=\mathbf{S}_{\vec{x}}A\mathbf{S}_{\vec{x}}$, where
\begin{equation}
 \mathbf{S}_{\vec{x}}=\exp[i\frac{\pi}{2}(b_{\vec{x}}^{\dagger}-a_{\vec{x}}^{\dagger})(b_{\vec{x}}-a_{\vec{x}})],
\end{equation}
and so $\beta s_{\vec{x}}\beta^{-1}$ is equivalent to conjugation by
\begin{equation}
 \exp[i\frac{\pi}{2}(b_{\vec{x}}^{\prime\dagger}-a_{\vec{x}}^{\dagger})(b^{\prime}_{\vec{x}}-a_{\vec{x}})],
\end{equation}
where $b^{\prime}_{\vec{x}}=\beta(b_{\vec{x}})$, which must be localized within the neighbourhood of $\vec{x}$ because $\beta$ is causal.
\end{proof}

\subsection{Representation by Qubits}
\label{sec:Representation by Qubits 2}
We have already seen that $\mathbf{S}_{\vec{x}}$ is local in the qubit representation with a sensible choice of ordering for the Jordan-Wigner Transformation.  To guarantee that the other unitaries are local in the qubit picture, however, we may need to apply the trick of adding Majorana fermions from appendix \ref{sec:Representing Local Fermionic Hamiltonians by Local Qubit Hamiltonians}.  We saw that $\beta s_{\vec{x}}\beta^{-1}$ is equivalent to conjugation by
\begin{equation}
\label{eq:17}
\exp[i\frac{\pi}{2}(b^{\prime\dagger}_{\vec{x}}-a^{\dagger}_{\vec{x}})(b^{\prime}_{\vec{x}}-a_{\vec{x}})].
\end{equation}
This has the form $e^{-iH_{\vec{x}}}$, where $H_{\vec{x}}$ is a self adjoint operator localized on the neighbourhood of $\vec{x}$, so we can apply the trick of adding Majorana fermions to get local unitaries in the qubit picture.  This is always possible because $b^{\prime\dagger}_{\vec{x}}=\beta(b^{\dagger}_{\vec{x}})$ must be a finite linear combination of odd powers of creation and annihilation operators from the neighbourhood of $\vec{x}$.

Note that the number of additional fermions we need to add per site is finite for causal evolutions.  

When we considered a finite number of fermion modes, we needed the state to be a $+1$ eigenstate of all the  operators $M_{\vec{x}\vec{y}}=im_{(\vec{x},\vec{y})} m_{(\vec{x},\vec{y})}$ we introduced.  Now, in the infinite case, we need to construct an analogous state.  Take any state of the original fermions and extend it to a state on the total system of physical and additional Majorana fermions $\rho$, with the property that all additional fermionic modes are empty.  Then define the state $\sigma$ by a family of states $\sigma_{\Lambda}$ on finite regions $\Lambda$:
\begin{equation}
\begin{split}
 \sigma_{\Lambda}(A) & =\rho(K^{\dagger}AK)\\
\textrm{with}\ K & =\prod_{\substack{\vec{x} \in\Lambda\\ \vec{y} \in\mathbb{Z}^d}}\frac{1}{\sqrt{2}}(m_{(\vec{x},\vec{y})}-im_{(\vec{y},\vec{x})}),
\end{split}
\end{equation}
where $A$ is localized on $\Lambda$.  Note that there are only finitely many $m_{(\vec{x},\vec{y})}$ for each $\vec{x}$ because the neighbourhood of $\vec{x}$ is finite.  Also, the ordering of the terms in the product is not critical.   The state $\sigma$ has the property that $\sigma(M_{(\vec{x},\vec{y})}A)=\sigma(AM_{(\vec{x},\vec{y})})=\sigma(A)$, (this can be seen from equation \ref{eq:70} at the start of the following section) so that the local fermionic unitaries augmented with the Majorana fermions are equivalent to the original unitaries on this state.  Furthermore, the results of a measurement on the physical fermions in the state $\sigma$ are the same as those from the same measurement on the original state.

To view this as a QCA, we map the fermionic system to a qubit lattice via the Jordan-Wigner Transformation. Given any sensible ordering of the infinite lattice sites (for example, starting at the origin and spiralling outwards filling progressively larger cubes), this is an isomorphism.  It will map every element of the fermionic C*-algebra to an element of the qubit C*-algebra, as the string of $Z$s that arise from the Jordan-Wigner Transformation for any $a_{\vec{x}}$ or $m_{(\vec{x},\vec{y})}$ will be finite.

So, by adding additional Majorana fermions, we have local qubit unitaries $V_{\vec{x}}$ implementing the local fermionic unitaries of equation \ref{eq:17} when in the state $\sigma$.  Note, however, that it is not necessarily true that $V_{\vec{x}}$ and $V_{\vec{y}}$ commute when the neighbourhood of $\vec{x}$ and the neighbourhood of $\vec{y}$ overlap.  But this does not matter.  We can implement them in a finite number of steps, where each step involves applying $V_{\vec{x}}$ unitaries that are localised on areas that do not overlap.  Note that the order in which they are applied does not matter because these operators commute when acting on the state $\sigma$.  So we have extended theorem \ref{th:2} to the infinite case.

\section{Preparing the Majorana state}
\label{sec:Preparing the Majorana state}
We want to prepare a state in the qubit picture that is a $+1$ eigenstate of $M_{(\vec{x},\vec{y})}=im_{(\vec{x},\vec{y})}m_{(\vec{y},\vec{x})}$.  To do this, first notice that, because $m_{(\vec{x},\vec{y})}^2=m_{(\vec{y},\vec{x})}^2=1$,
\begin{equation}
\label{eq:70}
M_{(\vec{x},\vec{y})}(m_{(\vec{x},\vec{y})}-im_{(\vec{y},\vec{x})})=(m_{(\vec{x},\vec{y})}-im_{(\vec{y},\vec{x})}).
\end{equation}
It follows that
\begin{equation}
 \prod_{\langle \vec{x}\vec{y}\rangle}\frac{1}{\sqrt{2}}(m_{(\vec{x},\vec{y})}-im_{(\vec{y},\vec{x})})\ket{\Omega}
\end{equation}
is a $+1$ eigenstate of all Majorana pairs $M_{(\vec{x},\vec{y})}$, where $\langle \vec{x}\vec{y}\rangle$ denotes pairs of sites where we add Majorana fermions.  The order of the product is irrelevant since any order will be a $+1$ eigenstate of the $M_{(\vec{x},\vec{y})}$ pairs.  Note that
\begin{equation}
\begin{split}
&\prod_{\langle\vec{x}\vec{y}\rangle}\frac{1}{\sqrt{2}}(m_{(\vec{x},\vec{y})}-im_{(\vec{y},\vec{x})})\ket{\Omega}\\
&=\prod_{\langle \vec{x}\vec{y}\rangle}\frac{1}{\sqrt{2}}(c^{\dagger}_{(\vec{x},\vec{y})}-ic^{\dagger}_{(\vec{y},\vec{x})})\ket{\Omega},
\end{split}
\end{equation}
where, $m_{(\vec{x},\vec{y})}=c^{\dagger}_{(\vec{x},\vec{y})}+c_{(\vec{x},\vec{y})}$, with $c_{(\vec{x},\vec{y})}\ket{\Omega}=0$.  This state can be created up to a phase by applying the unitaries
\begin{equation}
\label{eq:16}
\begin{split}
 &e^{i\frac{\pi}{2}\mathcal{B}_{(\vec{x},\vec{y})}}=\\
&\exp[i\frac{\pi}{2}(\frac{1}{\sqrt{2}}[c^{\dagger}_{(\vec{x},\vec{y})}-ic^{\dagger}_{(\vec{y},\vec{x})}]+\frac{1}{\sqrt{2}}[c_{(\vec{x},\vec{y})}+ic_{(\vec{y},\vec{x})}])]
\end{split}
\end{equation}
to $\ket{\Omega}$.  To see this, note that $\mathcal{B}_{(\vec{x},\vec{y})}^2=1$, which implies that
\begin{equation}
 e^{i\theta \mathcal{B}_{(\vec{x},\vec{y})}}=\cos(\theta)+i\sin(\theta)\mathcal{B}_{(\vec{x},\vec{y})}.
\end{equation}

We want to create the invariant state on qubits, but in the qubit representation the creation operators $c^{\dagger}_{(\vec{x},\vec{y})}$ still have those awkward strings of $Z$ operators, making this a non local unitary.  Still, the unitaries in equation \ref{eq:16} can be implemented efficiently by using a method presented in \cite{AL97} to deal with the strings of $Z$ operators.  First, consider all the qubits that the qubit representation of $c^{\dagger}_{(\vec{x},\vec{y})}$ acts on with a $Z$.  We can map the parity of these qubits to a flag qubit, which means that a single $Z$ acting on the flag qubit has the same effect as the string of $Z$s applied to the other qubits.  For example, with $r_i\in\{0,1\}$,
\begin{equation}
\begin{split}
 Z_0...Z_n\ket{r_0...r_n} & \ket{\sum_{j=0}^n r_j\bmod 2}=\\
 (-1)^{\sum_{j=0}^n r_j\bmod 2}\ket{r_0...r_n} & \ket{\sum_{j=0}^n r_j\bmod 2}=\\
 \ket{r_0...r_n}Z & \ket{\sum_{j=0}^n r_j\bmod 2}.
\end{split}
\end{equation}
So, after preparing flag qubits for $c_{(\vec{x},\vec{y})}^{\dagger}$ and $c^{\dagger}_{(\vec{y},\vec{x})}$, the qubit unitary in (\ref{eq:16}) is equivalent to a unitary on four qubits.  After this step, we need to reverse the operation preparing the flag qubits, but this and the original flag preparation can be done using only $n$ two qubit unitaries, where $n$ is the number of qubits we count the parity of.  For example, if $n=2$, we need only two steps:
\begin{equation}
\begin{split}
 & \ket{r_1r_2}\ket{0}_{f}\rightarrow \ket{r_1r_2}\ket{r_1}_{f} \rightarrow\\
 & \ket{r_1r_2}\ket{(r_1+r_2)\bmod 2}_{f},
\end{split}
\end{equation}
where $r_i\in\{0,1\}$ and the subscript $f$ denotes the flag qubit.  There is a constant number of qubits associated to each site because the number of Majorana fermion pairs we introduce per site does not grow with $\mathcal{N}$, so counting the parity takes $O(\mathcal{N})$ two qubit operations.  Therefore, preparing the $+1$ eigenstate of all Majorana terms can be done in time $O(\mathcal{N}^2)$.

Note that we could also use this method to apply the $U_B\mathbf{S}_{\vec{x}\mu}U_B^{\dagger}$ gates without introducing additional Majorana fermions, but this would mean an overhead of $O(\mathcal{N})$ for each gate.

Similarly, to create initial physical fermions states (essentially, by applying $a_{\vec{x}}^{\dagger}$), we can  use the same method to deal with strings of $Z$s appearing in the qubit picture, which will add an overhead of $O(\mathcal{N})$ steps.
\end{document}